\documentclass[runningheads]{PMS2016}

\usepackage[ansinew]{inputenc}
\usepackage[T1]{fontenc}
\usepackage[top=3cm,bottom=3.5cm,left=4cm,right=3.8cm,dvips]{geometry}
\usepackage{nopageno}

\usepackage{fancyhdr}

\usepackage{indentfirst}
\usepackage{sectsty}
\allsectionsfont{\usefont{OT1}{phv}{bc}{n}\selectfont \normalsize \bf }
\usepackage{lmodern} 

\usepackage[dvips]{graphicx}
\usepackage{subfigure}
\usepackage{url}
\usepackage{harvard}
\usepackage{amssymb}
\usepackage{epsfig}

\usepackage{amssymb,amsmath,graphicx}
\usepackage{tabularx}
\usepackage{psfrag,graphics,caption}
\usepackage{pgf}
\usepackage{epsfig}
\usepackage{color}
\usepackage{wasysym}
\usepackage{tikz}
\usepackage{dot2texi}
\usepackage{multirow}
\usetikzlibrary{decorations,arrows,shapes}

\def \afig#1#2#3 {\begin{figure}\{htbp\}
\begin{center} 
\mbox{\psfig{file=#1.eps,width=#3}}
\end{center}
\caption{#2}
\label{fig:#1}
\end{figure}}

\def \bfig#1#2 {\begin{figure}\{hhh\}
\begin{center} 
\mbox{\psfig{file=#1.eps}}
\end{center}
\caption{#2}
\label{fig:#1}
\end{figure}}

\def\bbbn{\rm I\!N}

\makeindex

\begin{document}

\pagenumbering{empty}
\pagestyle{headings}
\mainmatter


\title{Bounds and approximation results for scheduling coupled-tasks with compatibility constraints}

\author{R. Giroudeau\inst{1},  J.C K\"{o}nig \inst{1}, B. Darties\inst{2} and G. Simonin \inst{3}}

\institute{
LIRMM UMR 5506, 161 rue Ada 34392, Montpellier France\\
\email{$\{$rgirou,konig$\}$@lirmm.fr}
\and
LE2I UMR6306, Univ. Bourgogne Franche-Comt\'e, F-21000 Dijon, France\\
\email{benoit.darties@u-bourgogne.fr}
\and
Insight Centre for Data Analytics, University College Cork, Ireland\\
\email{gilles.simonin@insight-centre.org}
}

\maketitle

\index{B. Darties}\index{R. Giroudeau}\index{J.C K\"{o}nig}\index{G. Simonin}

\noindent


\begin{abstract}
This article is devoted to propose some lower and upper bounds for the coupled-tasks scheduling problem in presence of compatibility constraints  according to classical complexity hypothesis  ($\mathcal{P} \neq \mathcal{NP}$, $\mathcal{ETH}$). Moreover, we develop  an efficient polynomial-time approximation algorithm for the specific case for which the topology describing the compatibility constraints is a quasi split-graph.\newline
\textbf{Keywords:} coupled-task, compatibility graph, complexity, approximation.

 \end{abstract}
 \noindent

\section{Introduction, motivations, model}

We consider in this paper the coupled-task scheduling problem subject to compatibility constraints. The motivation of this model is related to data acquisition processes using radar sensors: a sensor emits a radio pulse (first sub-task $a_i$), and listen for an echo reply (second sub-task $b_i$). To make the notation less cluttered, the processing time of a sub-task will be denoted by $a_i$ instead of $p_{a_i}$ used in the theory of scheduling. Between these two instants (emission and reception), clearly there is an idle time $L_i$ due to the propagation, in both sides, of the radio pulse. A coupled-task $(a_i, L_i, b_i)$, introduced by \citeasnoun{Shapiro}, is a natural way to model such data acquisition. This model has been widely studied in several works,\texttt{ i.e.} \citeasnoun {BEKPTW09}. Other works proposed a generalization of this model by including compatibility constraints: scheduling a sub-task  during the idle time of another requires that both tasks are compatible.  The relations of compatibility are modeled by a compatibility graph $G$, linking pair of compatible tasks only. This model is detailed in \citeasnoun{sdgk11journal}. In previous works, we studied the complexity of scheduling coupled-tasks with compatible constraints under several parameters like the size of the sub-tasks or the class  of the compatibility graph~\cite{sgk13}. 

In this work, we propose original complexity and approximation results for the problem of scheduling \textit{stretched} coupled-task with compatibility constraints. A \textit{stretched} coupled-tasks $i$ is a coupled-task having both sub-tasks processing time and idle time equal to a triplet $(\alpha(i),\alpha(i),\alpha(i))$, where  $\alpha(i)$ is the $\textit{stretch factor}$ of the task $i$ - one can apply a stretch factor $\alpha(i)$ to a reference task $(1,1,1)$ to obtain $i$ -. 

The objective is to minimize the makespan $C_{max}$. 
The input of the problem is a collection of coupled-tasks $\mathcal{T}=\{t_1, t_2, \dots t_n$\}, a stretch factor function $\alpha : \mathcal{T}\rightarrow \bbbn$, and a compatibility graph $G_c=(\mathcal{T},E)$ where edge from $E$ link pairs of compatible tasks only.  When dealing with stretched coupled-tasks only,  a edge $\{x,y\} \in E$ exists if $\alpha(x) = \alpha(y)$ (then $x$ and $y$ can be scheduled together without idle time as the idle time of one task is employed to schedule the sub-task of the other, thus we can schedule sequentially $a_x, a_y, b_x, b_y$ - or $a_y, a_x, b_y, b_x$ - in $\frac{4\alpha(x)}{3}$ time units), or if $3\alpha(x) \leq \alpha(y)$ (then $x$ can be entirely executed during the idle time of $y$ \textit{i.e.} $a_y, a_x, b_x, b_y$ and scheduling both tasks requires  $3\alpha(y)$ time units). We note $\#(X)$ the number of different stretch factors in a set of tasks $X$, and we note $d_{G}(X)$ the maximum degree of any vertex $x\in X$ in a graph $G_c$. 

We use the well-known Graham notation \cite{GLLRK79} to define the problems presented in this paper. In this work, we propose new complexity and inapproximability results when the compatibility graph is a restricted  $1-stage~bipartite$ graph $G=(X,Y,E)$, \texttt{i.e.} a bipartite graph where edges are oriented from $X$ to $Y$ only. Then we show the problem is $\mathcal{NP}$-complete on a quasi-split graph $G=(G_X,G_Y,E)$\footnote{A quasi split graph is a connected graph $G=(G_X,G_Y,E)$, with $G_X$ a connected non-oriented graph (not complete)  and $G_Y$ a independent set. The other arcs are oriented from $X$ to $Y$ only.} even if $\#(V(G_X))=1$ and $\#(V(G_Y))=1$, but is $5/4$-approximable. 

\section{Complexity and approximation results}

\begin{theorem}
\label{bipartiinapprox}
  Deciding whether an instance of
   $1|\alpha, G_c =$  $1-stage-bipartite,$ $ \#(X)=2,\#(Y)=1,d_{G_c}(X) \in \{1,2\}, d_{G_c}(Y)\in \{3,4\} | C_{max}$ is a problem hard to approximate within $\frac{21-\rho^{\textsc{Max-3DM-2}}}{20} \leq \rho$, where $\rho^{\textsc{Max-3DM}}$ gives the upper bound for the \textsc{Max-3DM}. Since $\rho^{\textsc{Max-3DM-2}} \leq \frac{140}{141}$, we obtain $1+\frac{1}{2820}$.
   \end{theorem}

\begin{proof}

We prove first that the problem is $\mathcal{NP}$-complete via a polynomial-time reduction. Based on this reduction, we apply the gap-preserving reduction.

The proof is based on a reduction from the maximum \textsc{$3$ Dimensional Matching} (\textsc{Max-3DM}) \cite{np}: let $A$, $B$, and $C$ be three disjoint sets of equal size, with $n=|A|=|B|=|C|$, and a set $T \subseteq A \times B \times C$ of triplet, with $|T|=m$. The aim is to find a matching (set of mutually disjoint triplets) $T^* \subseteq T$ of maximum size. This problem is well known to be $\mathcal{NP}$-complete.
The restricted version of this problem in which each element of $A\cup B \cup C$ appears exactly twice is denoted \textsc{Max-3DM-2} and remains $\mathcal{NP}$-complete \cite{Chlebik}. In this restricted version, we have $m=2n$. 


We transform the instance of \textsc{Max-3DM-2} to an instance of $1|\alpha, G_c =  1-stage~bi$ $ partite,$ $\#(X)=2,\#(Y)=1,d_{G_c}(X) \in \{1,2\}, d_{G_c}(Y)\in \{3,4\} | C_{max}=63n- 3k (1-\epsilon)$ as follows: we define a set of tasks $X\cup Y$ and model the compatibility constraint with a graph $G_c=(X,Y,E)$. For each element $x_i \in A\cup B \cup C$, we add an \textit{item} coupled-task $x_i$ into $X$ with $\alpha(x_i)=1$. For each triplet $t_i \in T$, we add a \textit{box} coupled-task $t_i$ to $Y$ with $\alpha(t_i)=9$, and an \textit{item} coupled-task $t'_i$ with  $\alpha(t'_i)=2+\epsilon$. For each $t_i\in T$ and each $x_i \in t_i$,  we add the compatibility arc $(x_i, t_i)$ to $E$. We also add the  compatibility arc $(t'_i, t_i)$ to $E$. So, the set of $X$-tasks (resp. $Y$-tasks) are constituted by \textit{item} coupled-task $x_i$ and $t'_i$ (resp. \textit{box} coupled-task).

Clearly we have $m$ \textit{box} coupled-tasks (each with an idle time of $9$ units) of degree $4$ in $G_c$, $m$ \textit{item} coupled-tasks with stretch factor $2+\epsilon$ of degree $1$ in $G_c$, and $3n$ \textit{item} coupled-tasks with stretch factor $1$  of degree $2$ in $G_c$. Moreover $G_c$ is a bipartite graph. 
The reduction  is constructed in polynomial time. 


 It exists a schedule of length $63n- 3k (1-\epsilon)$ iff it exists a matching of size $k$ for \textsc{Max-3DM-2} instance.


\end{proof}

\enlargethispage{0.5cm}
Hereafter, we propose some negative results  concerning the existence of subexponential-time algorithms under the following complexity-theoretic hypothesis that is known as the  Exponential-Time Hypothesis (see \cite{Woeginger01} for a survey on exact algorithms for $\mathcal{NP}$-hard problems) for stretched coupled-tasks, and other ones previously studied.
 
 Recall first the {\sc Exponential-Time Hypothesis} (\cite{ImpagliazzoP01}, and \cite{ImpagliazzoPZ01}): there exists a constant $c >1$ such that there exists no algorithm for $3-$Satisfiability  that uses only $O(c^l)$ time where $l$  denotes the number of variables.

\begin{corollary}\label{ETH}
Assuming the Exponential-Time Hypothesis, there exists no algorithm with a worst-case running time that is subexponential in $n$ (the number of vertices),  \textit{i.e.}:
  
  \begin{enumerate}
  \item For the $1|a_i=b_i=p, L_i=2p,G_c | C_{max}$ problem   in $O(2^{o(n)})$ time
  \item For $1|a_i=a, b_i=b, L_i=a+b,G_c | C_{max}$ in $O(2^{o(n)})$ time
\item  $1|\alpha,G_c=1-bipartite| C_{max}$ in $O(2^{O(n)})$-time algorithm.
\end{enumerate}
\end{corollary}

\begin{proof}

\begin{enumerate}
  \item  For $1|a_i=b_i=p, L_i=2p,G_c| C_{max}$:  In  \cite{RooijNB13}, the authors proved that for {\sc Partition into triangles} on graphs of maximum degree four, there is no algorithm with a worst-case running  time  $O(2^{o(n)})$ that is subexponential in $n$.
 
  Therefore, we transform a {\sc Partition into triangles} instance with $n$ vertices and $m$ edges into an equivalent instance $G_c$  for bounded degree at most four. Since the transformation is linear (see \cite{sdgk11journal}) the result holds.
  
  \item For the problem $1|a_i=a, b_i=b, L_i=a+b,G_c | C_{max}$: In \cite{LokshtanovMS11} the authors proved that for {\sc Hamiltonian path} there is no   $O(2^{o(n)})$-time algorithm. As the same way as previously the transformation is linear (see \cite{sdgk11journal}).
\item  $1|\alpha,G_c=1-bipartite | C_{max}$: In \cite{ChenJZ14}, the authors proved that for {\sc Max 3DM}, there is no $O(2^{O(n)})$-time algorithm, therefore this result is transposed to the scheduling problem using the first part of the proof of Theorem \ref{bipartiinapprox}.
\end{enumerate}
\end{proof}

\begin{theorem}
Scheduling stretched coupled task in presence of a quasi split graph is a $\mathcal{NP}-$complete problem even if $\#(V(G_X)) = 1$ and $\#(V(G_Y)) = 1$
\end{theorem}

\begin{proof}
The proof is based on a reduction from a variant of the well-know $\mathcal{NP}$-complete {\sc Partition into triangles}. This problem consists to ask if  the vertices of a graph $G=(V,E)$, with $|V|=3q, q \in \bbbn$, can be partitioned into $q$ disjoints sets $T_1, T_2, \ldots, T_q$, each containing exactly three vertices, such that for 
each $T_i=\{u_i,v_i,w_i\}, 1 \leq i \leq q$, all three of the edges $\{u_i,v_i\}, \{u_i,w_i\}, \{w_i,v_i\}$ belong to $E$.

The problem {\sc Partition into triangles} remains $\mathcal{NP}$-complete even if the graph $G$ can be partitioned into three sets with the same size, $A,~B$ et $C$ such that each set is an independent  set \cite{morandini}. The polynomial-time transformation is based on this variant.
Let $G = (A \cup B \cup C,E)$ be an instance of the variant of {\sc Partition into Triangles}. We consider the  split-graph $G'=(A \cup B,C,E')$ obtained as follows:


$\forall v \in A$ (resp. $B$), we create a vertex $A_v$ (resp. $B_v$) with processing time $(1,1,1)$. Moreover, $\forall v \in C$ we create a task $C_v$ with  processing time $(4,4,4)$.
The edges between $A$ and $B$ remained the same as the $G'$ whereas the edge between $A \cup B$ and $C$ are oriented.
Finally in order to have a connected graph, we add two news vertices (resp. one) $z_0$ and  $z_1$ (resp. $z_2$ with processing time equal to $(1,1,1)$ (resp. $(4,4,4)$). We add  edges between $z_0$ to $A_v$ (resp. $z_1$ to $B_v$). Lastly, we add the three edges $(z_0,z_2)$, $(z_1, z_2)$ and $(z_0, z_1)$.

Notice that the graph $A_v \cup B_v$ form a bipartite graph. The problem is clearly in $\mathcal{NP}$. It exists a positive solution for the variant of  {\sc Partition into triangles} iff a valid schedule of length $12 \times (|C|+1)$  exists. It is sufficient to execute the two tasks $A_v$ and $B_{v'}$ in four units of time into a task $C_u$.


\end{proof}


\begin{theorem}
\label{approxsplit}
The problem is  $5/4$-approximable on quasi split-graph where $\#(V(G_Y)) = 1$.
\end{theorem}
\enlargethispage{0.5cm}
\begin{proof}
W.l.o.g., we suppose that the processing time of $X$-tasks (resp. $Y$-tasks) is $(1,1,1)$ (resp. $\alpha(y_i)$).  Indeed, if $\alpha(x) >1$, we put $\alpha(y_i)=\lfloor \frac{\alpha(y_i)}{\alpha(x)} \rfloor$ and $\alpha(x)=1$ .

\textbf{Algorithm:} we  transform the problem into an oriented maximum flow-problem between $G_X$ and $G_Y$ with two sources $s$ and $t$, with $\omega(s,x)=\omega(x,y)=1$ and $\omega(y,t)=\lfloor \frac{\alpha(y_i)}{3\alpha(x)}\rfloor, \forall y_i \in Y, \forall x \in XG_Y$ where $\omega(i,j)$ is the capacity of an arc $(i,j)$ . After the computation of a maximum flow $F$ of value $f$,  for the uncovered remaining $X$-tasks a maximum $M$-matching ($|M|=m$) is applied. The schedule consists in processing first,  the $Y$-tasks with $X$-tasks inside. The $M$-tasks are executed after. Lastly, we schedule $s$ isolated-tasks. The length of schedule given by the algorithm is $C_{max} \leq  \sum_{y_i \in Y} 3\alpha(y_i)+4m+3s$ with $2m+s+f=n=|X|$ and $ \sum_{y_i \in Y} 3\alpha(y_i) \geq 9f$. In similar way, the optimal length is $C^*_{max} \geq \sum_{y_i \in Y} \alpha(y_i)+4m^*+3s^*$. We suppose that in $Y$-tasks where are $p^*$-edges processed and $r^*$ isolated-tasks, then we obtain $2(p^*+m^*)+r^*+s^*=n$, $p^*+r^* \leq f$, and $\sum_{y_i \in Y} \alpha(y_i) \geq 12p^*+9r^*$. In the worst-case, the $p^*$-edges are split  into two tasks (so $p^*$ news tasks are added to $s^*$), and also the matched-edges  are split (for each $m^*$ edges one task is executed into the $Y$-task, instead of  one of $r^*$-tasks). Therefore, $2m^*$ tasks are added to the $s$-value.
In the worst case, we have $m^*=r^*$, $s=s^*+p^*+2r^*$ and $m=0$. In such case,  $C_{max} \leq 12p^*+9r^*+3s^*+3p^*+6r^*$ and $C^*_{max}=12p^*+9r^*+4r^*+3s^*$. Thus $\rho \leq \frac{15p^*+15r^*+3s^*}{12p^*+13r^*+3s^*} \leq \max(5/4,15/13,1)=5/4$.



\textbf{Tightness:} it exists an example for the $C^*_{max}=36$, and for the heuristic $C_{max}=45$. Consider the graph: three triangles $(x_1,x_2,y_1)$, $(x_3,x_4,y_2)$, and  $(x_5,x_6,y_1)$. We add the edges $(x_2,y_3$, $(x_3,y_1)$ and $(x_5,y_2)$. 
The optimal solution consists in executing the $X$-tasks into the $Y$-tasks; whereas the heuristic leads the solution in which three $X$-tasks are processed after the $Y$-tasks.
\end{proof}


\bibliographystyle{agsm}

\end{document}